\documentclass[11pt]{amsart}

\usepackage{hyperref}
\usepackage{amsfonts}
\usepackage{amssymb}
\usepackage{amsmath}
\usepackage{amsthm}
\usepackage{latexsym}
\usepackage{color}
\usepackage{amscd}
\usepackage[all,cmtip]{xy}
\usepackage{graphicx}
\usepackage{enumitem} 
\usepackage[margin=1.5in]{geometry}
\usepackage{esint} 
\usepackage{epsfig}

\allowdisplaybreaks

\numberwithin{equation}{section}

\newtheorem{theorem}{Theorem}

\newtheorem{lemma}[theorem]{Lemma}

\theoremstyle{definition}

\theoremstyle{remark}
\newtheorem{remark}{Remark}[section]

\newcommand\R{{\mathbb R}}
\newcommand\T{{\mathbb T}}

\newcommand\Z{{\mathbb Z}}

\newcommand{\CC}{\mathbb{C}}

\newcommand{\LL}{\mathbf{L}}

\newcommand{\PP}{\mathbf{P}}

\newcommand{\cA}{\mathcal A}

\newcommand{\cD}{\mathcal D}

\newcommand{\cH}{\mathcal H}

\newcommand{\cL}{\mathcal L}

\newcommand{\cR}{\mathcal R}
\newcommand{\cS}{\mathcal S}

\newcommand{\cU}{\mathcal U}

\def\M{\mathcal{M}}
\def\H{\mathcal{H}}

\def\LL{{\bf L}}
\def\CC{{\bf C}}
\def\PP{{\bf P}}

\def\eps{{\varepsilon}}

\newcommand{\dd}{{\, \mathrm d}}

\newcommand{\ds}{\displaystyle}
\newcommand{\fa}{\forall \,}
\newcommand{\sk}{\smallskip}
\newcommand{\mk}{\medskip}

\newcommand{\rd}{{\rm d}}
\newcommand{\var}{\varepsilon}

\let\oldmarginpar\marginpar
\renewcommand\marginpar[1]{\-\oldmarginpar[\raggedleft\footnotesize #1]%
{\raggedright\footnotesize #1}}

\def\signec{\bigskip \begin{center} {\sc Eric
      Carlen\par\vspace{3mm}
      Department of Mathematics\par
      Rutgers University\par
      110 Felinghuysen Rd., Piscataway N.J. 08541 U.S.A.
      \par\vspace{3mm} e-mail:}
    \tt{carlen@math.rutgers.edu} \end{center}}

\def\signjl{\bigskip \begin{center} {\sc Joel
      Lebowitz\par\vspace{3mm}
      Department of Mathematics \& Department of Physics\par
      Rutgers University\par
      110 Felinghuysen Rd., Piscataway N.J. 08541 U.S.A.
      \par\vspace{3mm} e-mail:}
    \tt{lebowitz@math.rutgers.edu} \end{center}}

\def\signcm{\bigskip \begin{center} {\sc Cl\'ement
      Mouhot\par\vspace{3mm}
      DPMMS, Centre for Mathematical Sciences\par
      University of Cambridge\par
      Wilberforce road, Cambridge CB3 0WA, U.K.
      \par\vspace{3mm} e-mail:}
    \tt{c.mouhot@dpmms.cam.ac.uk} \end{center}}

\def\signre{\bigskip \begin{center} {\sc Rafaelle
      Esposito\par\vspace{3mm}
      Dipartimento di Matematica\par
      Universit\`a di l'Aquila\par
      Coppito, 67100 AQ, ITALY
      \par\vspace{3mm} e-mail:}
    \tt{esposito@roma2.infn.it} \end{center}}

\def\signrm{\bigskip \begin{center} {\sc Rossana
      Marra\par\vspace{3mm}
      Dipartimento di Fisica and Unit\`a INFN\par
      Universit\`a di Roma Tor Vergata\par
      00133 Roma, ITALY
      \par\vspace{3mm} e-mail:}
    \tt{marra@roma2.infn.it} \end{center}}

\begin{document}

\title[Uniqueness of the NESS for a $1$d BGK model in kinetic
theory]{Uniqueness of the Non-Equilibrium Steady State for a $1$d BGK model in kinetic theory}


\author{E. Carlen, R. Esposito, J. Lebowitz, R. Marra \&
  C. Mouhot}

\date{\today}

\begin{abstract} We continue our investigation of kinetic models of a
  one-dimensional gas in contact with homogeneous thermal reservoirs
  at different temperatures. Nonlinear collisional interactions
  between particles are modeled by a so-called BGK dynamics which
  conserves local energy and particle density. Weighting the nonlinear
  BGK term with a parameter $\alpha\in [0,1]$, and the linear
  interaction with the reservoirs by $(1-\alpha)$, we prove that for
  some $\alpha$ close enough to zero, the explicit spatially uniform
  non-equilibrium stable state (NESS) is \emph{unique}, and there
  are no spatially non-uniform NESS with a spatial density $\rho$ belonging to $L^p$ for any $p>1$.
  We also show that for all
  $\alpha\in [0,1]$, the spatially uniform NESS is \emph{dynamically
    stable}, with small perturbation converging to zero exponentially
  fast. 

\end{abstract}

\maketitle
\tableofcontents

\section{Introduction}

This paper is a contribution to the theory of non-equilibrium steady
states (NESS), of open systems in the particular context of kinetic
theory. The understanding of NESS, their properties, uniqueness or
lack thereof and stability or lack thereof, represents a challenge in
mathematical physics due to the fact that the dynamics are nonlinear, 
non-Markovian and  the absence of an entropy principle. Our
main result is a uniqueness and stability theorem for the NESS in a
simple nonlinear model.

\subsection{The model} We briefly describe the sort of underlying
particle model that would lead to the type of kinetic equation that we
study here.  It consists of a gas of particles on the one-dimensional
torus $\T$, that interact only through binary energy conserving
collisions, however we also suppose that there are two types of scatterers
distributed on the torus according to some Poisson distribution, as in
a Lorentz model, except that each scatterer has a temperature, $T_1$
or $T_2$ depending on its type, and a certain radius of interaction,
so that when a gas particle travels across the interaction interval, a
Poisson clock runs and if it goes off, the particle assumes a new
velocity chosen at random according to the Maxwellian distribution for
the temperature of the scatterer.

In an appropriate scaling limit, the net effect of the background
scatterers is to produce two uniform thermal reservoirs. Whatever the
speed of a gas particle, its rate of interaction with the reservoirs
depends only on the density of the scatterers, again in an appropriate
limit in which their intervals of interaction are unlikely to
overlap.  The kinetic equation that one would expect to arise from
such a model in such a limit would be of the type
\eqref{bgkthermo} below, except that one might expect a Kac-Povzner
type collision kernel, also known as a ``soft-spheres'' kernel \cite{Ce,PPS}. Our work concerns the
kinetic equation itself, and not its rigorous derivation from an
underlying particle system, although the brief description of such a
system that we have given hopefully illuminates the physical context
of our model.

We are concerned with the existence and uniqueness of NESS for our
system. We make a further simplification, and model the gas particle
collisions with a BGK collision kernel \cite{BGK,Pe89,PePu93}. This will render the existence
of NESS trivial, but the uniqueness is still a challenging problem,
and we shall only prove part of what we conjecture to be true.

These considerations bring us to the following one dimensional kinetic
model:
\begin{equation}\label{bgkthermo}
  \partial_t f + v \partial_x f
  = \alpha \mathcal M_f + (1-\alpha) \rho_f
  \left( \frac{M_{T_1}+M_{T_2}}{2} \right) -f
\end{equation}
where $\alpha \in [0,1]$, $f = f(t,x,v)$, $x \in \T$, $v \in \R$,
and
\begin{equation}\label{maxdef}
  M_{T_i}(v) := \frac{e^{-\frac{|v|^2}{2T_i}}}{\sqrt{2\pi T_i}}, \quad
  \mathcal M_f(t,x,v)
  := \rho_f(t,x) \frac{e^{-\frac{|v|^2}{2T_f(t,x)}}}{\sqrt{2\pi T_f(t,x)}},
\end{equation}
with
\begin{equation}\label{RTdef}
  \begin{cases}\ds
    \rho_f(t,x) := \int_\R f(t,x,v)\dd v,  \\[3mm] \ds
    P_f(t,x) := \int_\R v^2 f(t,x,v) \dd v =\rho_f(x) T_f(t,x),
  \end{cases}
\end{equation}
being the {\em spatial density} and {\em
  pressure} corresponding to $f$.  $T_f$ is then the {\em temperature} corresponding to $f$,  and $T_1$ and $T_2 \in (0,+\infty)$
being the two temperatures of the reservoirs.

The linear terms on the right in \eqref{bgkthermo} that are multiplied
by $(1-\alpha)$ model the interaction of particles with two
reservoirs, both acting everywhere in space.  Each time a particle
interacts with one of the reservoirs, it velocity is replaced by a new
velocity selected at random from the corresponding Maxwellian
distribution. We have taken both of these Maxwellians to have zero
mean velocity which is natural for static reservoirs.

The nonlinear term on the right in \eqref{bgkthermo}  that is
multiplied by $\alpha$ represents the effect of collisions between
particles. The collision term $\mathcal M_f$ is of  the BGK type
(see~\cite{BGK}), except that as usual in one dimensional kinetic
models, it conserves only mass and energy, not momentum. Indeed, binary
collisions that conserve both energy and momentum are trivial in one
dimension: only an exchange of velocities is possible. For this
reason, the mean velocity of ${\mathcal M}_f$ is zero.

The term $-f$ on the right in \eqref{bgkthermo} is the {\em loss term}
corresponding to both interactions with the reservoirs and other
particles:  after such interactions, particles vacate their pre-interaction state.

Without loss of generality, we choose units in which the torus has
unit volume and there is unit total mass:
\[
  \T = [-1/2,1/2] \quad \mbox{ and } \int_{-1/2} ^{1/2} \int_{\R} f
  (x,v) \dd x \dd v=1.
\]

\subsection{Previous results}

In our previous papers \cite{CLM,CELMM}, we have studied related
issues for  related models. In \cite{CLM} we proved the
existence of spatially homogeneous non-equilibrium steady states and
exponential convergence to them for related spatially homogeneous
models, but with more realistic collision mechanisms, and also in
higher dimensions.  In \cite{CELMM} we studied the exponential rate of
convergence to steady state for a non spatially homogeneous equation
of the type \eqref{bgkthermo} but with a modified collision mechanism
that permitted the equation to be interpreted as the Kolmogorov
forward equation for a non-stationary Markov process: we replaced the
space-dependent local temperature $T_f(t,x)$ by the global temperature
\[
  \mathcal T_f(t) := \int_{\T \times \R} v^2 f \dd x\dd v
\]
of $f$ that depends only on time. We were then able to apply Doeblin's 
method \cite{Var} to prove the exponential convergence.  The use of Doeblin's method
to study linear dynamical models originates with \cite{BL,LB}.

The rigorous study of NESS  for nonlinear kinetic equations remains very challenging. 
One problem that has been studied by several authors is the Boltzmann equation in a slab with
different temperatures on the two walls, with and without external forces. At this level of generality, one cannot always expect
a unique NESS -- there may be a symmetry breaking transition, such as the onset of Rayleigh-Bernard flow.  Even without
external forces, existence of NESS for the slab problem is a highly non-trivial, and existing results \cite{AN,Ce2} do not yield provide any information on uniqueness or non-uniqueness.

More recently, the Boltzmann equation in more general domains and with non-isothermal boundary conditions has been investigated \cite{EGKM} where it
is proved that when the temperature on the boundary is sufficiently close to constant, then there is an NESS that is 
close to the uniform Maxwellian for the mean boundary temperature, and in a small neighborhood of this Maxwellian, there is no other NESS. 
However, it is not known that there are not other NESS further away, no matter how small the non-zero temperature difference may be. 

\subsection{Question studied and conjecture}

Observe that the reservoirs will tend to damp out any mean velocity
since $M_{T_1}$ and $M_{T_2}$ have zero mean velocity. Likewise,
${\mathcal M}_f$ has zero mean velocity at each $x$, so the collision
gain term too will  tend to damp out any mean velocity.

Therefore, if $f=f(v)$ is any spatially homogeneous steady state,
$\int_\R v f {\rm d}v =0$. Moreover, the time and space homogeneity yield
$\partial_t f =0$ and $v \partial_x f =0$. Finally, multiplying both
sides of \eqref{bgkthermo} by $v^2$ and integrating over $x$ and $v$
shows
\[
  0 = (\alpha -1)T_f  + (1-\alpha)\frac{T_1+T_2}{2}, \qquad T_f :=
  \int_\R v^2 f(v) \dd v.
\]
Thus, the constant temperature in any spatially homogeneous steady
state $f$ must be $T_\infty := (T_1+T_2)/2$ if $\alpha \not =1$. Then
for any spatially homogeneous steady state $f$,
$\mathcal M_f = M_{(T_1+T_2)/2}$, and \eqref{bgkthermo} reduces to
$$
\alpha M_{\frac{T_1+T_2}{2}} + (1-\alpha) \rho_f \left(
  \frac{M_{T_1}+M_{T_2}}{2} \right) -f =0.
$$
Therefore the unique spatially homogeneous steady state is given
by
\begin{equation}\label{stt}
  f_\infty := \alpha M_{T_\infty} +
  (1-\alpha) \frac{M_{T_1}+M_{T_2}}{2}, \quad
  T_\infty = \frac{T_1+T_2}{2}. 
\end{equation}
Observe that $f_\infty$ is {\em not} Maxwellian as soon as
$\alpha \not = 1$.

If $\alpha = 0$, the term $\mathcal M_f$ is not present, the only
spatially homogeneous steady state is
$f_\infty = \frac12(M_{T_1} + M_{T_2})$, and the equation
\eqref{bgkthermo} is linear. It can be interpreted as the forward
equation of a Markov process and in \cite{CELMM} we used probabilistic
methods to prove that this steady state is unique and is approached
exponentially fast.  Hence, for $\alpha= 0$, there are no steady
states that are spatially inhomogeneous.

Next, consider the case $\alpha = 1$: there are no thermal reservoirs
and energy is conserved. There is a one-parameter infinite family of
steady states, namely $M_T$ for all $T>0$.  Moreover, if $f_0$ is such
that
\[
  \int_{\T \times \R} v^2 f_0(x,v) \dd x \dd v = T, \qquad
  \int_{\T \times \R} f_0(x,v) \ln f_0(x,v) \dd x \dd v < +\infty,
\]
and $f(t,x,v)$ is the solution of \eqref{bgkthermo} with initial datum
$f_0$, then
\[
  H \left(f(t,\cdot,\cdot) | M_T \right) =
  \int_{\T \times \R} f(t,x,v) \ln
  \frac{f(t,x,v)}{M_T(v)} \dd x \dd v
\]
decreases monotonically to zero, and is stationary only when
$f=M_T$. It follows that $M_T$ is the unique steady state among
solutions with second moment equal to $T$ and finite entropy, and thus
every steady state for $\alpha =1$ with finite second moment and
entropy is spatially homogeneous (and equal to $M_T$).

The question that motivates this paper is the study of the NESS in the
intermediate region $\alpha \in (0,1)$. We conjecture the following:
\mk

\noindent {\bf Conjecture (Uniqueness of the NESS for the BGK model
  with reservoirs).}  {\em For all $\alpha\in [0,1]$, the
  non-equilibrium steady state of \eqref{bgkthermo} is unique, regardless of the temperature difference, 
  spatially homogeneous and stable under  perturbations. We also
  expect this conjecture to hold in higher dimensions $x \in \T^d$,
  $v \in \R^d$.}

\subsection{Main results}

We give  a partial answer to this conjecture, showing
that it is satisfied when $\alpha$ is small enough. We first prove the
uniqueness:
\begin{theorem}\label{main}
  For all $T_1,T_2$, there is an explicitly computable $\alpha_0> 0$ such that for all
  $\alpha \in [0,\alpha_0)$, every steady state solution $f_\infty$ of
  \eqref{bgkthermo} that belongs to $L^1(\T \times \R)$, has finite
  second moment and is such that $\rho \in L^p(\T)$ for some $p>1$,
  is constant in $x$.
\end{theorem}

\begin{remark}
  Our method would also apply in higher dimensions {\em provided}
  there was non-trivial spatial dependence in only one direction on
  the torus, say the $x_1$ coordinate.  The decomposition between odd
  and even parts in the next section should be then modified by
  splitting along $v_1$ only. The rest of the analysis would be
  similar.
\end{remark}

We also prove the stability under perturbation for all $\alpha \in
[0,1]$. For this, we introduce the (real)  Hilbert Space $\cH_\alpha^1$ with inner product
\begin{equation}\label{sob}
\langle f,g\rangle_{\cH_\alpha^1} = \int_{\T \times \R}(f(x,v) (1- \partial_x^2) g(x,v)) \frac{1}{f_{\alpha,\infty}}\dd x \dd v\ .
\end{equation}
\begin{theorem}\label{stab}
  For all $\alpha \in [0,1]$, the spatially homogeneous steady state
  described above is asymptotically stable under perturbation in
  $\cH_\alpha^1$. Small perturbations
  decay exponentially fast in time in this space.
\end{theorem}
Theorem~\ref{stab} shows that if for some $\alpha> \alpha_0$ there do
exist non-uniform steady states, they do not arise as a branch
bifurcating off the family of spatially homogeneous steady state
solutions.

\subsection{Plan of the paper}

In Section~\ref{sec:relations}, we establish some useful relations on
the moments of any given NESS and introduce a decomposition between
odd and even parts. In Section~\ref{Sec2}, we prove lower and upper
pointwise bounds on local density and temperature of any given
NESS. In Section~\ref{Sec3}, we explain the contraction mapping
argument; it is in this section that we use that $\alpha$ is close
to zero. Finally in Section~\ref{sec:linearised}, we prove a local
stability result of the spatially homogeneous steady states for all
$\alpha \in [0,1]$. 

\section{Preliminaries: properties and decompositions of NESS}
\label{sec:relations}

\subsection{Zero momentum and constant pressure}

A partial result supporting our conjecture is that the pressure is
independent of $x$, as well as the momentum
$m_f(x) := \rho_f(x)u_f(x)$, the latter being zero:
\begin{lemma}\label{conpres}
  Let $f(x,v)$ be a probability density on $\T \times \R$ such that
  $v^2 f(x,v)$ is integrable, and suppose that $f(x,v)$ solves in a
  weak sense the equation
  \begin{equation}\label{conprA}
    v\partial_x f(x,v) = F(x,v) - f(x,v)
  \end{equation}
  where $F(x,v)$ is a measurable function such that $(1+|v|)F(x,v)$
  is integrable and 
  \begin{equation}\label{conprB}
    \forall \, x \in \T, \quad \int_\R F(x,v) \dd v = \rho_f(x)
    \quad{\rm and}\quad
    \int_\R v F(x,v) \dd v = 0.
  \end{equation}
  Then the pressure is constant and the momentum is zero: 
  \[
    \left\{
      \begin{array}{l} \ds
        P_f(x) = \int_\R v^2 f(x,v) \dd v = P_\infty
        \in \R\\[3mm] \ds 
        m_f(x) = \int_\R v f(x,v) \dd v = \rho_f(x)u_f(x)= 0.
      \end{array}
    \right.
  \]
\end{lemma}

\begin{remark} Evidently, Lemma~\ref{conpres} applies to any finite
  energy NESS of our equation.
\end{remark}

\begin{proof}
  Integrating both sides of~\eqref{conprA} in $v$ yields
  \[
    \frac{\rd}{\rd x} m_f(x) = \rho_f(x)  -
    \rho_f(x) = 0.
  \]
  This proves that $m_f(x)$ is a constant $m_\infty \in \R$. Now
  multiplying both sides of~\eqref{conprA} by $v$ and integrating
  in $v$ yields
  \begin{equation}\label{consP}
    \frac{\rd}{\rd x} P_f(x) = -m_f(x) = - m_\infty.
  \end{equation}
  Integrating both sides of \eqref{consP} in $x$ shows that
  $m_\infty=0$, and $P_f(x) = P_\infty \in \R$ is constant.
\end{proof} 

\begin{remark}
  The proof of Lemma~\ref{conpres} takes advantage of the dimension
  being one. In higher dimension, the argument above would only show
  that $m_f$ is a divergence free vector field, but not necessarily
  constant.  We shall make further use of the one dimensionality of
  the model when proving pointwise bounds on the NESS.
\end{remark}

\subsection{Higher moments}

Multiplying the steady-state equation
\begin{equation}\label{SS}
 v \partial_x f = \alpha \mathcal M_f + (1-\alpha) \rho_f
  \left( \frac{M_{T_1}+M_{T_2}}{2} \right) -f
\end{equation}
by $v^2$,  and integrating in $v$ yields
\[
  \frac{\rd}{\rd x} \int_\R v^3 f(x,v)\dd v =
  (1-\alpha)\left(\frac{T_1+T_2}{2}\right)\left(\rho_f(x)-1\right),
\]
since by Lemma~\ref{conpres}, $P_f(x) = (T_1+T_2)/2$.  Next,
multiplying \eqref{SS} by $v^3$ and integrating yields
\[
  \frac{\rd}{\rd x} \int_\R v^4 f(x,v)\dd v = -\int_\R v^3 f(x,v)\dd
  v.
\]
Combining the last two equations yields
\begin{equation}\label{ST}
  -\frac{\rd^2}{\rd x^2} \int_\R v^4 f(x,v)\dd v =
  (1-\alpha)\left(\frac{T_1+T_2}{2}\right)\left(\rho_f(x) -1\right).
\end{equation}
Since the right hand side integrates to zero, we have
\[
  \int_\R v^4 f(x,v) \dd v - \int_{\T \times \R} v^4 f(x,v) \dd x \dd
  v= (1-\alpha)\left(\frac{T_1+T_2}{2}\right) \int_\T \psi(x-y)
  (\rho_f(y) -1){\rm d} y
\]
where
\[
  \psi(x) = \sum_{k\neq 0} \frac{e^{2\pi i kx}}{4\pi^2k^2} \qquad{\rm
    so \ that}\quad |\psi(x)| \leq \frac{1}{12}.
\]
It follows that
\begin{equation}\label{ST2}
  \left|    \int_\R v^4 f(x,v)\dd v   - \int_{\T \times \R} v^4
    f(x,v)\dd x \dd v \right| \leq  (1-\alpha)\left(\frac{T_1+T_2}{12}\right).
\end{equation}
In particular, for $\alpha$ close to $1$, $ \int v^4 f(x,v)\dd v$ is
nearly constant; its average is
\begin{equation}\label{ST4}
  \int_{\T \times \R} v^4 f(x,v)\dd v{\rm d}x =
  3\left[\alpha \left(\frac{T_1+T_2}{2}\right)^2 +
    (1-\alpha)\frac{T_1^2 + T_2^2}{2}\right]
\end{equation}
and thus the spatial fluctuations in $\int v^4 f(x,v)\dd v$ are a
small fraction of the mean for large temperatures.

\begin{lemma}\label{lowerby4}
  Let $f$ be a solution to \eqref{SS} such that
  \[
    \int_{\T\times \R}(1+v^2)f(x,v) \dd x\dd v < \infty
  \]
  and recall that $P_\infty := (T_1+T_2)/2$.  Then \eqref{ST2} is valid, and
  the spatial density $\rho_f$ satisfies
  \begin{equation}\label{ST3}
    \rho_f(x) \geq 
    \frac{1}{ 3(2-\alpha) + \frac{(1-\alpha)}{6P_\infty}}.
  \end{equation}
\end{lemma}

\begin{proof}
  By the Cauchy-Schwarz inequality and Lemma~\ref{conpres},
  \[
    \fa \, x \in \T, \quad P_\infty = P_f(x) = \int_\R v^2f(x,v)\dd v
    \leq \left( \int_\R v^4 f(x,v)\dd v \right)^{1/2}\rho_f^{1/2}(x),
  \]  
  so that
  \[
    \rho_f(x) \geq  P_\infty^2 \left( \int_\R v^4 f(x,v)\dd v \right)^{-1}.
  \]  
  
  From \eqref{ST4}, we have the bounds
  \[
    \begin{cases}\ds
      \int_{\T \times \R} v^4 f(x,v)\dd x \dd v & \leq
      3(2-\alpha)P_\infty^2 \\
      \ds \sup_{x \in \T} \int_{\R}
      v^4 f(x,v)\dd v & \leq 3(2-\alpha)P_\infty^2 + \frac{(1-\alpha)
        P_\infty}{6}.
    \end{cases}
  \]
  Combining bounds yields the result.
\end{proof}

\subsection{Splitting between odd and even parts and a wave-like
  system}

We split a given steady state $f$ into even and odd parts $f = E + O$
with respect to the $v$ variable. The steady state
equation~\eqref{bgkthermo} can be rewritten as:
\begin{equation}\label{pair}
  \left\{ 
    \begin{array}{l}
      \ds v \partial_x E = - O \\[3mm]
      \ds v \partial_x O = F_\alpha - E
    \end{array}
  \right.
\end{equation}
where 
\begin{equation}\label{Fdef}
  F_\alpha(x,v):= \alpha \mathcal M_f(x,v) + (1-\alpha) \rho_f(x)
  G(v),
  \quad G := \left( \frac{M_{T_1}+M_{T_2}}{2} \right).
\end{equation}
Combining the two equations in (\ref{pair}), we obtain
\begin{equation}\label{ellip}
  \left(1 - v^2\partial_x^2\right)E =  F_\alpha.
\end{equation}
Note that for each $v\neq 0$, the operator $(1 - v^2\partial_x^2)$ is
invertible with a bounded kernel. The equation (\ref{ellip})
conveniently and efficiently expresses the iterated effects of
velocity averaging on the steady state, or rather on its even part.

\begin{lemma}\label{one}
  For each density $\rho$ on $\T$, there is at most one NESS $f$
  such that $\rho= \rho_f$.
\end{lemma}

\begin{proof}
  We have seen that for any NESS, $P_\infty = (T_1+T_2)/2$, and
  then by Lemma~\ref{conpres}, $\rho_f T_f = P_\infty$, so that
  ${\mathcal M}_f$ is determined by $\rho$.  It follows that
  $F_\alpha$ is determined by $\rho$, and then since \eqref{ellip} is
  uniquely solvable, the uniqueness of $f$ follows.
\end{proof}

The formal solution of \eqref{ellip2} is
\begin{equation}\label{ellip2}
  E = \left(1 - v^2\partial_x^2\right)^{-1}F_\alpha 
\end{equation}
and can be written in terms of an explicit Green's function.
Integrating in $v$ yields
\begin{equation}\label{ellip3}
  \rho_f(x)  = \int_\R \left[ \left(1 - v^2\partial_x^2\right)^{-1}F_\alpha \right] \dd v.
\end{equation}

\begin{lemma}\label{fixP}
  Let $T_1,T_2>0 $ and $\alpha \in (0,1)$ and
  $P_\infty = (T_1+T_2)/2$. For any probability density $\rho=\rho(x)$ on
  $\T$, define $T(x) = P_\infty/\rho(x)$ and 
  \begin{equation}\label{Psidef}
    \begin{cases}\ds
      {\mathcal M}[\rho](x,v)
       & \ds := \frac{\rho(x)}{\sqrt{2\pi T(x)}}
      e^{- \frac{v^2}{2T(x)}} = \frac{\rho^{3/2}(x)} {\sqrt{2\pi P_\infty}}
      e^{- \frac{v^2\rho(x)}{2P_\infty}} \\[3mm] \ds F_{\alpha}[\rho](x,v)
       & \ds :=
      \alpha {\mathcal M}_\rho(x,v) + (1-\alpha)\rho(x) G(v)
      \\[3mm] \ds
      \Psi_\alpha[\rho](x)
       & \ds := \int_\R \left[ \left(1 -
          v^2\partial_x^2\right)^{-1}F_\alpha[\rho] \right] \dd v
    \end{cases}
  \end{equation}
  with $G$ defined as in \eqref{Fdef}.  Then for all $\rho$, $\Psi_\alpha[\rho[$ is a probability density on $\T$, and  $\rho$ is the spatial
  density of some NESS $f$ if and only if $\rho = \Psi_\alpha[\rho]$,
  and in this case the unique such NESS is given in terms of $\rho$ by
  \eqref{ellip2B} and \eqref{ellip2C} below.
\end{lemma} 

\begin{proof}
Let $\rho$ be any density.  Then $F_{\alpha}[\rho]$ is a probability density on $\T\times \R$.
Since $(1 + v^2\partial_x^2)^{-1}$ preserves both integrals and positivity,  $(1 + v^2\partial_x^2)^{-1}F_{\alpha}[\rho]$ is also a 
a probability density on $\T\times \R$, and hence its marginal, $\Psi_\alpha[\rho]$, is a probability density on $\T$. 

  Next, suppose that $\rho(x) = \rho_f(x)$ for some NESS $f(x,v)$. By
  Lemma~\ref{conpres}, ${\mathcal M}[\rho](x,v) = {\mathcal M}_f(x,v)$,
  and therefore $F_{\alpha}[\rho](x,v)$ is given in terms of $f$ by
  \eqref{Fdef}.  Then 
  $E(x,v)$, given by
  \eqref{ellip2},  is a probability density on $\T\times \R$. $E(x,v)$ is the even part of $f(x,v)$ and finally $O(x,v)$, the odd part of $f(x,v)$, is given
  by the first equation in \eqref{pair}.  Then by \eqref{ellip3},
  $\rho = \Psi_\alpha[\rho]$.

Finally,  suppose that $\rho(x)$ is a probability density on $\T$ such
  that $\rho = \Psi_\alpha[\rho]$.  Define $F_{\alpha}[\rho]$ by
  \eqref{Psidef}, and then define $E_{\alpha}[\rho](x,v)$ by
  \begin{equation}\label{ellip2B}
    E_{\alpha}[\rho] := \left(1 - v^2\partial_x^2\right)^{-1}F_{\alpha}[\rho],
  \end{equation}
  and then define
  \begin{equation}\label{ellip2C}
    O_{\alpha}[\rho] = - v\partial_xE_{\alpha}[\rho] \qquad{\rm and}\qquad
    f_\rho(x,v) = E_{\alpha}[\rho](x,y) + O_{\alpha}[\rho](x,y).
  \end{equation}
  Then
  \begin{eqnarray*}
    v\partial_x f_{\rho}(x,v)
    &=& -O_{\alpha}[\rho] -  v^2\partial_x^2 E_{\alpha}[\rho]\\
    &=& -O_{\alpha}[\rho] -E_{\alpha}[\rho] + (1-v^2\partial_x^2) E_{\alpha}[\rho]\\
    &=& F_{\alpha}[\rho] - f_\rho.
  \end{eqnarray*}
  Lemma~\ref{conpres} applies to this equation, and we conclude that
  $\rho_{f_\rho} T_{f_\rho} = P_\infty$. The fixed point equation
  $\rho = \Psi_\alpha[\rho]$ implies $\rho_{f_\rho} = \rho$ and
  ${\mathcal M}[\rho] = {\mathcal M}_{f_\rho}$.  This shows that
  $f_\rho$ is an NESS for our equation, and concludes the proof that
  $\rho$ is the spatial density of an NESS if and only if it is a
  fixed point of $\Psi_\alpha$.
\end{proof} 

It follows from Lemma~\ref{fixP} that our conjecture would be proved
if we could show that the constant density is the only fixed point of
$\Psi_\alpha$ for all $\alpha\in [0,1]$. We prove this for
sufficiently small $\alpha$ in the next section.

\section{Pointwise bounds on the moments of the NESS}\label{Sec2}

\subsection{Preliminaries}

We define the standard Fourier series of an integrable  function $r=r(x)$ on the
torus $\T = [-1/2,1/2]$ by
\[
  \hat r(k) := \int_0 ^1 r(x) e^{-2i\pi kx} \dd x, \quad k \in \Z,
\]
and we recall the inversion formula (when, say, the Fourier
modes $(\hat r(k))_{k \in \Z}$ are absolutely summable)
\[
  r(x) = \sum_{k \in \Z} \hat r(k) e^{2i\pi kx}.
\]

Define $\varphi_v(x)$ the fundamental solution to the Laplace equation
\begin{equation*}
  \left(1-(v \partial_x)^2\right)^{-1} \varphi_v(x) = \delta_0(x)
\end{equation*}
on the circle $\T = [-1/2,1/2]$. This fundamental solution is
explicit:
\begin{equation*}
  \varphi_v(x) = \sum_{m \in \Z} \phi_v(x+m), \quad \phi_v(x)
  := \frac{1}{2|v|} e^{-\frac{|x|}{|v|}}
\end{equation*}
and its formula in Fourier is
\begin{equation*}
  \hat \varphi_v(k) = \int_{-1/2} ^{1/2} e^{-i2\pi kx} \varphi_v(x)
  \dd x = \frac{1}{1+(2\pi)^2 v^2 k^2}, \quad k \in \Z.
\end{equation*}

The following bounds will be useful: $\varphi_v \in L^p(\T)$ for all
$p\in [1,+\infty]$, and
\begin{equation}\label{Lpbnd}
  \|\varphi_v\|_{L^p(\R)} = \left(\frac1p\right)^{\frac1p}
  \left(\frac{1}{2|v|}\right)^{\frac{p-1}{p}}, \ p \in [1,+\infty),
  \qquad \|\varphi_v\|_{L^\infty(\R)} = \left(\frac{1}{2|v|}\right).
\end{equation}
and it satisfies the lower bound 
\begin{equation}\label{Ppbnd}
  \fa x,y \in \T, \quad \varphi_v(x-y)  \geq \frac{1}{2|v|} e^{-\frac{1}{|v|}}.
\end{equation}

\subsection{Lower bound on the NESS}

We have already proved a uniform lower bound on the spatial density of
any NESS in Lemma~\ref{lowerby4}. We now prove a stronger result: a
uniform lower bound holds for every density in the range of
$\Psi_\alpha$.
\begin{lemma}[Pointwise lower bound]\label{lower}
  Let $\alpha \in [0,1]$ and let $\rho$ be any probability density on
  $\T$ of the form $\rho = \Psi_\alpha[\rho_0]$ for a probability
  density $\rho_0$ on $\T$. Then
  \begin{equation}\label{plb}
    \fa \, x \in \T, \quad \rho(x) \geq r_\infty
    := \frac{(1-\alpha)}{4\sqrt{e}} \int_{1 \leq |v|\leq
      2} \left( \frac{M_{T_1}+M_{T_2}}{2} \right) \dd v.
  \end{equation}
\end{lemma}


\begin{proof}
  Define $F_\alpha[\rho_0]$ in terms of $\rho_0$ using~\eqref{Psidef}.  Then
  \begin{equation}\label{plb1}
    \rho(x)= \int_{\T \times \R} \varphi_v(x-y) \, F_\alpha[\rho_0](y,v) \dd y
    \dd v.
  \end{equation}
  The operator $(1-(v \partial_x)^2)^{-1}$ preserves positivity, and
  since $F_\alpha[\rho_0] \geq(1-\alpha) \rho_0 G$, we obtain
  \begin{eqnarray*}
    \rho(x)
    &\geq&  (1-\alpha) \int_\R  G(v) \left( \int_{\T} \varphi_v(x-y) \,
           \rho_0(y) \dd y \right) \dd v\\
    &\geq&  (1-\alpha) \int_{ 1 \leq |v| \leq 2}  G(v) \left( \int_{\T} \varphi_v(x-y) \,
           \rho_0(y) \dd y \right) \dd v.
  \end{eqnarray*}
  
  The kernel $\varphi_v(x-y)$ is bounded below by
  $e^{-1/|v|}/(2|v|)$. The function defined by
  $t>0 \mapsto t^{-1}e^{-t^{-1}}$ vanishes as $t$ approaches zero or
  infinity, is maximum at $t=1$ and then decreases as $t$ increases,
  and has the value $1/(2 \sqrt e)$ at $t=2$. This implies
  $$\varphi_v(x-y) \geq \frac{1}{4\sqrt e} $$
  for all $x,y$ whenever $1 \leq |v| \leq 2$, which concludes the
  proof.
\end{proof}

\subsection{Upper bound on the NESS}

\begin{lemma}[Gain of integrability of $\Psi_\alpha$]
  Let $\alpha \in [0,1]$ and $r \in [0,1)$, and let
  $\rho_0 \in L^{1+r/2}(\T)$ a probability density. Then
  $\rho := \Psi_\alpha[\rho_0]\in L^{1+r}(\T)$, and there are
  $A_r, B_r >0$ depending only on $r$ and degenerating as $r \to 1$
  such that
  \begin{equation}\label{STU}
    \int_\T \rho^{1+r}\dd x \leq \alpha A_r \int_\T \rho_0^{1+r/2}\dd x + B_r\ .
  \end{equation}

  As a consequence, if $\rho=\rho_0 \in L^{1+r/2}(\T)$ is a fixed
  point of $\Psi_\alpha$, there is a constant $K_{r,\alpha} >0$
  depending only on $r$ and $\alpha$ and monotone increasing in
  $\alpha$ such that
  \[
    \int_\T \rho^{1+r}\dd x \leq K_{r,\alpha}.
  \]
\end{lemma}

\begin{remark}
  Note that the constant $K_{r,\alpha}$ is independent of
  $\| \rho\|_{L^{1+r/2}(\T)}$. That is, knowing only that
  $\int_\T \rho^{1+r/2}\dd x$ is finite, we obtain a universal bound
  on $\int_\T \rho^{1+r}\dd x$.
\end{remark}

\begin{proof}
  Define again $F_\alpha[\rho_0]$ in terms of $\rho_0$ using
  \eqref{Psidef} so that 
  \begin{equation}\label{plb1}
    \rho(x)= \int_{\T \times \R} \varphi_v(x-y) \, F_\alpha[\rho_0](y,v) \dd y
    \dd v.
  \end{equation}
  Multiply (\ref{plb1}) by $\rho^{r}(x)$ and integrate in $x$:
  \begin{align*}
    \int_{\T} \rho^{1+r}(x) \dd x
    & = \int_{\T} \rho^{r}(x) \left(  \int_{\T \times
      \R}\varphi_v(x-y) \, F_\alpha[\rho_0](y,v) \dd y \dd v \right) \dd x\\
    & = \int_{\T \times \R}  \left( \int_{\T}
      \varphi_v(x-y)\, \rho^{r}(x) \dd x   \right)  F_\alpha[\rho_0](y,v) \dd y \dd v.
  \end{align*}
  Equation~\eqref{Lpbnd} implies  
  \[
    \int_{\T} \varphi_v(x-y)\, \rho_0^{r}(x) \dd x \ \leq
    \|\rho_0\|^ r _{L^1(\T)}\|\varphi_v\|_{L^{\frac{1}{1-r}}(\T)} \leq \frac{1}{(2|v|)^{r}}.
  \]
  Therefore
  \begin{eqnarray*}
    \int_{\T} \rho^{1+r}(x) \dd x
    &\leq&   \int_{\T \times \R}   |v|^{-r} F_\alpha[\rho_0](y,v) \dd y \dd v\\
    &\leq&   \int_{\T \times \R}   |v|^{-r} \Big[\alpha {\mathcal M}_{\rho_0}(y,v)   + (1-\alpha)\rho_0(y) G(v)\Big] \dd y \dd v\\
    &\leq&  \alpha \int_{\T \times \R}   |v|^{-r} {\mathcal
           M}_{\rho_0}(y,v) \dd y \dd v +  \int_\R |v|^{-r}G(v)\dd v.
   \end{eqnarray*}
   Now using the definition of ${\mathcal M}_{\rho_0}$ and Lemma~\ref{conpres}: 
   \begin{equation*}
     \int_{\T \times \R}   |v|^{-r} {\mathcal M}_{\rho_0}\dd y \dd v
     \leq  \int_{\T}  \left(\int_\R  |v|^{-r}
            \frac{\rho_0^{3/2}(y)}  {\sqrt{2\pi P_\infty}} e^{- \frac{v^2\rho_0(y)}{2P_\infty}}\dd v \right)\dd y
   \end{equation*}
   and making the change of variable $w = (\rho_0(y)/P_\infty)^{1/2}v$,
   \[
     \int_\R |v|^{-r} \frac{\rho_0^{3/2}(y)} {\sqrt{2\pi P_\infty}}
     e^{- \frac{v^2\rho_0(y)}{2P_\infty}}\dd v = \rho_0(y)^{1+r/2}
     \frac{1}{\sqrt{2\pi P_\infty^r}}\int_\R |w|^{-r} e^{-\frac{|w|^2}{2}}\dd
     w.
   \]
   This yields \eqref{STU} with
   \[
     A_r := \frac{1}{\sqrt{2\pi P_\infty^r}}\int_\R |w|^{-r}
     e^{- \frac{|w|^2}{2}}\dd w \quad{\rm and}\quad B_r := \int_\R
     |v|^{-r}G(v)\dd v.
   \]

   Now suppose that $\rho = \rho_0 \in L^{1+r/2}(\T)$ is a fixed point
   of $\Psi_\alpha$. We have from the inequality and H\"older's
   inequality that
 \[ 
 \int_\T \rho^{1+r}\dd x \leq \alpha A_r \int_\T \rho^{1+r/2}\dd x +
 B_r \leq \alpha A_r \left(\int_\T \rho^{1+r}\dd
   x\right)^{\frac{1+r/2}{1+r}} + B_r.
 \]
 Since the exponent $(1+r/2)/(1+r)$ on the right is less than one,
 this proves that there is a constant $K_{r,\alpha} >0$ depending only
 on $r$ and $\alpha$ such that
 $\int_\T \rho^{1+r}\dd x \leq K_{r,\alpha}$. 
 \end{proof}

 \begin{lemma}[Pointwise upper bound]
   \label{upper}
   Let $\alpha \in [0,1]$ and let $\rho$ be any probability density on
   $\T$ that is a fixed point of $\Psi_\alpha$, and such that
   $\rho\in L^p$ for some $p>1$. Then $\rho$ satisfies the pointwise
   upper bound $\rho(x) \le R_\infty$ where $R_\infty< \infty$ only
   depends on the total energy, $\alpha$ and on $p$, and is monotone
   increasing in $\alpha$.
\end{lemma} 

\begin{proof}
  In the case where $p \in (1,7/4)$, finitely many applications of the
  previous lemma will yield, for any $q \in (p,2)$, a bound
  \[
    \int_\T \rho^q(x)\dd x \leq C_{r,\alpha}
  \]
  for some finite constant $C_{r,\alpha}$ depending only on $r$ and
  $\alpha$. We deduce that for all $p >1$, the following control holds
  \[
    \int_{\T}\rho^{7/4}(x)\dd x \leq C_{r, \alpha}.
  \]
  
  We return to (\ref{plb1}) and expand $F_\alpha$ to write
  \begin{equation}\label{plb5}
    \rho (x)= \alpha \int_{\T \times \R} \varphi_v(x-y) \mathcal
    M_\rho(x,v) \dd y \dd v + (1-\alpha) \int_{\T \times \R}
    \varphi_v(x-y) \rho (y) G(v) \dd y \dd v.
  \end{equation}
  Observe that Lemma~\ref{lower} implies 
  \begin{equation*}
    {\mathcal M}_{\rho}(x,v) =  \frac{\rho^{3/2}(x)}  {\sqrt{2\pi
        P_\infty}}
    e^{- \frac{v^2\rho(x)}{2P_\infty}} \leq
    \frac{\rho_0^{3/2}(x)}  {\sqrt{2\pi P_\infty}}
    e^{- \frac{r_\infty v^2}{2P_\infty}}.
  \end{equation*}
  which yields 
  \begin{multline}\label{plb6}
    \int_{\T \times \R} \varphi_v(x-y) \mathcal
    M_f(x,v) \dd y \dd v
    \leq\\
    \frac{1}  {\sqrt{2\pi P_\infty}}\int_\R \exp\left[-\frac{r_\infty v^2}{2 P_\infty}\right]  
    \left( \int_{\T} \varphi_v(x-y) \rho ^{3/2}(y) \dd y \right)
    \dd v.
  \end{multline}
  We apply H\"older's inequality with conjugate exponents $p = 7/6$
  and $q = 7$ to obtain
  \begin{align*}
    \left( \int_{\T} \varphi_v(x-y) \rho ^{3/2}(y) \dd y \right)
    & \leq \|\varphi_v\|_{L^7(\T)}
      \|\rho ^{3/2}\|_{L^{7/6}(\T)} \\
    & \le \|\varphi_v\|_{L^7(\T)}
      \|\rho \|_{L^{7/4}(\T)}^{3/2} \\
    & \le |v|^{-6/7} \|\rho  \|_{L^{7/4}(\T)}^{3/2} \\
    & \le C_{r,\alpha} |v|^{-6/7}
  \end{align*}
  where in the final step we have the used estimate~\eqref{Lpbnd} with
  $p=7$.  Using this in (\ref{plb6}) and noting that
  \[
    \int_\R \exp\left[-\frac{r_\infty v^2}{2 P_\infty}\right]
    |v|^{-6/7} \dd v < \infty,
  \]
  we deduce a universal upper bound on
  $\int_{\T \times \R} \varphi_v(x-y) \mathcal M_f(x,v) \dd y \dd v$.
  
  The term
  $ \int_{\T \times \R} \varphi_v(x-y) \rho(y) G(v) \dd y \dd v$ is
  bounded using H\"older inequality and ~\eqref{Lpbnd}:
  \[
    \int_{\T} \varphi_v(x-y) \rho(y)\dd y \leq
    \|\varphi_v\|_{L^{7/3}(\T)}\|\rho\|_{L^{7/4}(\T)} \le C_{r,\alpha}
    |v|^{-4/7}.
  \]
  The two last inequalities combined with~\eqref{plb5} imply the
  pointwise bound on $\rho$.
\end{proof}

\section{The contraction mapping argument}\label{Sec3}

\subsection{Setting of the argument}
Recall the relation $\rho_f T_f = P_\infty$ and
\begin{equation*}\label{Fdef2}
  \begin{cases} \ds
    F_\alpha[\rho] = \alpha \mathcal \rho M_{\frac{T_\infty}{\rho}}
    + (1-\alpha) \rho G, \\[3mm] \ds
    G := \left( \frac{M_{T_1}+M_{T_2}}{2} \right) \\[3mm] \ds
    \Psi_\alpha[\rho] := \int_\R (1 -
    v^2\partial_x^2)^{-1}F_\alpha[\rho] \dd v.
  \end{cases}
\end{equation*}
The local density of any steady state must be a fixed point of
$\Psi_\alpha$.
When $\alpha= 0$, the map $\Psi_0$ is linear, and a consequence of the
spectral analysis of the next Section~\ref{sec:linearised} is that it is
strictly contractive in $H^1$ or $L^2$ norms. To extend it to small
positive $\alpha$, we make use of the {\it a-priori} bounds proved in
the previous Section~\ref{Sec2}.  

\subsection{The contraction estimate}

\begin{lemma}\label{CMT}
  Given any $\var \in (0,1)$, define
  \[
    \mathcal C_\var := \left\{\rho \in L^2(\T) \ :\ \int_{\T}
      \rho(x) \dd x = 1 \ \mbox{ and } \ 0<\var < \rho < 1/\var
    \right\}.
  \]

  Then $\Psi_\alpha(\mathcal C_\var) \subset L^2(\T)$ for all
  $\alpha \in [0,1]$ and there are $\alpha_\var,\delta_\var \in (0,1)$
  depending on $\var$ such that for all $\alpha \in [0,\alpha_\var)$:
  \[
    \forall \, \rho_1, \rho_2 \in \mathcal C_\var, \quad \|
    \Psi_\alpha(\rho_1) - \Psi_\alpha(\rho_2) \|_{L^2(\T)} \leq
    (1-\delta_\var) \| \rho_1 - \rho_2 \|_{L^2(\T)}.
  \]
\end{lemma}

\begin{proof}
  Recall that in our normalization the global conserved quantities
  satisfy
  \[
    \rho_\infty := \int_{\T \times \R} f(x,v) \dd x \dd v = 1, \quad
    T_\infty:= \frac{1}{\rho_\infty} \int_{\T \times \R} v^2 f(x,v) \dd x \dd v =
    \frac{T_1+T_2}{2}.
  \]

  The fact that $\Psi_\alpha(\mathcal C_\var) \subset L^2(\T)$ is
  straightforward and we only prove the contraction property. We
  linearize the map $\Psi_\alpha$ around a profile
  $\bar \rho \in \mathcal C_\var$ with global mass $1$ and global
  temperature $T_\infty$. The local temperature is
  $\bar T(x) = T_\infty/ \bar \rho(x)$. 
  We write the fluctuation
  \[
    \rho = \bar \rho + \sigma \ \mbox{ with } \ \sigma \in L^2(\T)
    \mbox{ and } \int_{\T} \sigma(x) \dd x =0.
  \]
  
  The functional derivative of $\Psi_\alpha$ is: 
\begin{align*}
  D\Psi_\alpha[\bar \rho] \cdot \sigma
  & = \alpha \int_\R \left[ 1 - v^2 \partial^2_x
    \right]^{-1} \left( \frac{3\sigma}{2} M_{\frac{T_\infty}{\bar \rho}}\sigma
    - \frac{v^2 \bar \rho \sigma }{2T_\infty} M_{\frac{T_\infty}{\bar \rho}}  \right) \dd v \\
  & \quad + (1-\alpha) \int_\R \left[ 1 - v^2 \partial^2_x
    \right]^{-1} \sigma \, G(v) \dd v.
\end{align*}
We estimate by duality for $\sigma_1, \sigma_2 \in L^2(\T)$:
\begin{align*}
  \left\langle \sigma_2, D\Psi_\alpha[\bar \rho] \cdot \sigma_1 \right\rangle
  & = \alpha  \int_{\T \times \R} \left( \left[ 1 - v^2 \partial^2_x
    \right]^{-1} \sigma_2 \right) \left( \frac{3\sigma_1}{2}
    M_{\frac{T_\infty}{\bar \rho}} - \frac{v^2\bar \rho \sigma_1}{2T_\infty}
    M_{\frac{T_\infty}{\bar \rho}} \right) \dd x \dd v \\ 
  & \quad +  (1-\alpha) \int_{\T \times \R} \left( \left[ 1 - v^2 \partial^2_x
    \right]^{-1} \sigma_2 \right) G(v) \sigma_1(x) \dd x \dd v\\
    & =: D_1 + D_2. 
\end{align*}
\sk

Let us study the first term $D_1$. Using the controls
$\bar \rho \in \mathcal C_\var$ we deduce
\begin{equation*}
  \left[ \int_{\T} \left( \frac{3}{2} M_{\frac{T_\infty}{\bar \rho}} +
      \frac{v^2\bar \rho}{2T_\infty}
      M_{\frac{T_\infty}{\bar \rho}}  \right)^2 \sigma(x)^2 \dd x
  \right]^{1/2} \le C_\var e^{-
    \nu_\var v^2} \| \sigma \|_{L^2(\T)}
\end{equation*}
for some constant $C_\var, \nu_\var >0$ depending on $\var$, and therefore 
\begin{align*}
  & \int_{\T \times \R} \left( \left[ 1 - v^2 \partial^2_x
  \right]^{-1} \sigma_2 \right) \left( \frac{3}{2}
    M_{\frac{T_\infty}{\bar \rho}} - \frac{v^2\bar \rho}{2T_\infty}
    M_{\frac{T_\infty}{\bar \rho}} \right) \sigma_1(x) \dd x \dd v \\
  & \quad \le C_\var \| \sigma_1 \|_{L^2(\T)} \int_\R \left( \int_{\T} \left( \left[ 1 - v^2 \partial^2_x
  \right]^{-1} \sigma_2 \right)^2 \dd x \right)^{1/2} e^{-
  \nu_\var v^2} \dd v \\ 
  & \quad \le C_\var \| \sigma_1 \|_{L^2(\T)} \int_\R \left(
    \int_{\T} \sigma_2 ^2 \dd x \right)^{1/2}  e^{-
    \nu_\var v^2} \dd v \\
  & \quad \le C_\var \| \sigma_1 \|_{L^2(\T)} \| \sigma_2 \|_{L^2(\T)}
\end{align*}
where we have used that $\|\varphi_v\|_{L^1(\T)} =1$. We conclude
that
\begin{equation*}
  D_1 \le \alpha C_\var \| \sigma_1 \|_{L^2(\T)} \| \sigma_2 \|_{L^2(\T)}.
\end{equation*}
\sk

Let us now study the second term $D_2$. We Fourier transform it in
$x$:
\begin{equation*}
  \int_{\T \times \R} \left( \left[ 1 - v^2 \partial^2_x
    \right]^{-1} \sigma_2 \right) G(v) \sigma_1(x) \dd x \dd v 
  =  \sum_{k \in \Z} \int_\R \frac{\hat \sigma_2(k)}{\left[ 1 + (2\pi)^2 v^2 k^2 \right]}
  G(v) \hat \sigma_1(k) \dd v.
\end{equation*}
The perturbation $\sigma$ has zero mass, hence $\hat \sigma(0) = 0$
and
\begin{align*}
  \left| \int_{\T \times \R} \left( \left[ 1 - v^2 \partial^2_x
  \right]^{-1} \sigma_2 \right) G(v) \sigma_1(x) \dd x \dd v
  \right| & \le \sum_{k \in \Z,  \ k \not =0} \int_\R \frac{\hat
            \sigma_2 (k)}{\left[ 1 + (2\pi)^2v^2 k^2 \right]} 
            G(v) \hat \sigma_1(k) \dd v \\ 
          & \le \left( \int_\R \frac{G(v)}{1+(2\pi)^2 v^2} \dd v \right) \sum_{k
            \in \Z, \ k \not = 0} \hat \sigma_1(k) \hat \sigma_2 (k) \\ 
          & \le (1-\delta_G) \| \sigma_1 \|_{L^2(\T)} \| \sigma_2\|_{L^2(\T)}
\end{align*}
with $\delta_G \in (0,1)$, where we have used that
\begin{equation*}
  \left( \int_\R \frac{G(v)}{1+(2\pi)^2 v^2} \dd v \right) < \left(
    \int_\R G(v) \dd v \right) =1.
\end{equation*}

Therefore the operator $D \Psi_\alpha[\bar \rho]$ is bounded
from $L^2(\T) \to L^2(\T)$ with the operator norm bounded by 
\begin{equation*}
  ||| D \Psi_\alpha[\bar \rho] ||| = \sup_{\| \sigma_1\|_{L^2(\T)} \le
    1, \ \| \sigma_2 \|_{L^2(\T)} \le 1} \left\langle \sigma_2,
    D\Psi_\alpha[\bar \rho] \cdot \sigma_1 \right\rangle \le \alpha C_\var + (1-\alpha) (1-\delta_G).
\end{equation*}
For $\alpha$ small enough we deduce
$||| D\Psi_\alpha[\rho] ||| < (1-\delta_\var)$ with
$\delta_\var \in (0,1)$. Finally, since the set $\mathcal C_\var$ is
convex, the mean value theorem gives, for
$\rho_1, \rho_2 \in \mathcal C_\var$:
\begin{align*}
  \| \Psi_\alpha(\rho_1) - \Psi_\alpha(\rho_2) \|_{L^2(\T)} 
  & = \left\| \int_0 ^1
    D\Psi_\alpha[(1-t)\rho_1 + t \rho_2] \cdot (\rho_2-\rho_1) \dd t
    \right\|_{L^2(\T)} \\
  & \le  (1-\delta_\var) \| \rho_1 - \rho_2 \|_{L^2(\T)}
\end{align*}
which shows the contraction property for the nonlinear map. 
\end{proof}

\subsection{Proof of the main Theorem~\ref{main}}

We now prove Theorem~\ref{main}. By Lemma~\ref{lower} and
Lemma~\ref{upper}, there is a $\var> 0$, depending only on $T_1$ and
$T_2$ such that every steady state of \eqref{bgkthermo} with
$\alpha < 1/2$ belongs to the set $\mathcal C_\var$, defined in
Lemma~\ref{CMT}. Then by Lemma~\ref{CMT}, there is an $\alpha_0$ with
$0 < \alpha_0 \leq 1/2$ such that $\Psi_\alpha$ is a strictly
contractive mapping on this convex set $\mathcal C_\var$, and hence
there is a unique fixed point in $\mathcal C_\var$. Since there is
always one spatially uniform steady state, it is the unique steady
state.

\section{Perturbative stability of the spatially uniform NESS}
\label{sec:linearised}

\subsection{Strategy}
In this section we investigate the perturbative stability of the
spatially homogeneous NESS
\begin{equation}\label{finfa}
  f_{\infty,\alpha} := \alpha M_{T_\infty} + (1-\alpha) G(v), \quad 
  G(v) := \frac{M_{T_1}+M_{T_2}}{2}, \quad T_\infty := \frac{T_1+T_2}{2}.
\end{equation}

When the reservoirs have different temperatures and are coupled to the
system, that is for $\alpha \in (0,1)$, there is transfer of heat
through collisions from the hot reservoir to the cold reservoir, and
there is no detailed balance; i.e., time reversal invariance is broken
in the steady state. This is reflected in the fact that the linearized
operator is a non-self-adjoint operator on
$\H_\alpha =L^2(f_{\infty,\alpha}^{-1})$ for $\alpha \in (0,1)$, as we shall
see.

Nonetheless, we shall prove that the linearized collision operator
still satisfies a \emph{microscopic coercivity} inequality
(see~\eqref{diss1}), expressing the dissipative nature of the linearized evolution on the orthogonal  complement of the null space of the generator. This fact is striking since we do not derive it, through
linearization, from a nonlinear entropy principle, which is the usual source of such inequalities.  In our non-equilibrium setting, there is no analog of the $H$-Theorem, and therefore we must prove it
by direct analysis of linearized collision operator.

Once the microscopic coercivity is proven, we can prove that our
system is {\it hypocoercive} by a variety of methods, and we briefly describe two of these.  
\subsection{Linearization around a spatially homogeneous NESS}

Consider densities $f$ that are close to $f_{\infty,\alpha}$ with
fluctuations denoted
\begin{equation}\label{hdef}
  \begin{cases}\ds
    h(x,v) & := f(x,v) - f_{\infty,\alpha}(v) \\[3mm]
    \rho(x) & \ds = 1 + \sigma(x) \qquad{\rm with}\qquad \sigma(x):=\int_\R h(x,v) \dd v,\\[3mm]
    P(x) & \ds = T_\infty + \tau(x) \qquad{\rm with}\qquad \tau(x):=\int_\R
    v^2h(x,v) \dd v,
  \end{cases}
\end{equation}
The fluctuations of the local density and local pressure have zero
mean:
\begin{equation}\label{mloc2}
  \int_{\T} \sigma(x) \dd x =0  \qquad{\rm and}\qquad   \int_{\T}
  \tau(x) \dd x =0.
\end{equation}

Consider the weighted $L^2$ Hilbert space of {\em real valued functions}  defined by the norm
\begin{equation}\label{normdef}
  \|h\|_{\H_\alpha}^2 = \int_{\T\times \R} |h(x,v)|^2
  \frac{1}{f_{\infty,\alpha}(v)} \dd x \dd v.
\end{equation}

We expand $\M_f - \M_{f_{\infty,\alpha}}$ to first order in terms of
$h$, $\sigma$, $\tau$: 
 \begin{align}\label{taylor}
   & \M_f(x,v) - \M_{f_{\infty,\alpha}}(v)
   = \frac{(1+\sigma)^{3/2}(x)}{\sqrt{2\pi
     (T_\infty+\tau(x))}}e^{-\frac{v^2(1+\sigma(x))}{2(T_\infty+\tau(x))}}
     - \frac{1}{\sqrt{2\pi T_\infty}}e^{-\frac{v^2}{2T_\infty}}\nonumber\\
   &\qquad \approx \left( \frac32 - \frac{v^2}{2T_\infty}\right)M_{T_\infty}(v)\sigma(x)
     + \left( -\frac{1}{2T_\infty}+ \frac{v^2}{2T_\infty^2}\right)M_{T_\infty}(v)\tau(x) \nonumber\\
   &\qquad = M_{T_\infty}(v)\sigma(x)
     + \frac{1}{2}  \left(\frac {v^2}{T_\infty} - 1\right)
     M_{T_\infty}(v)\left(\frac{1}{T_\infty} \tau(x) - \sigma(x)\right).
 \end{align}

 The fluctuation $h = f - f_{\infty,\alpha}$ satisfies the equation
 \begin{equation}\label{bgktherB}
   \partial_t h + v \partial_x h =  \alpha(\mathcal \M_f - \M_{f_{\infty,\alpha}}) +(1-\alpha)\sigma G  - h.
 \end{equation}
 To first order we obtain the linearized equation
 \begin{equation}
   \label{eq:linh}
   \partial_t h + \mathcal S h = \mathcal L_\alpha h
 \end{equation}
 with the free streaming operator $\mathcal S := v \partial_x$ and
 the linearized collision operator
 \begin{equation}\label{BGKlin}
   \cL_\alpha h(x,v) := \sigma(x)f_{\infty,\alpha}(x,v)
   + \frac{\alpha}{2} \left(\frac{\tau(x)}{T_\infty}  - \sigma(x)\right)
   \left(\frac {v^2}{T_\infty} - 1\right)  M_{T_\infty}(v) - h(x,v).
 \end{equation}
 
 Note that both $\cL_\alpha$ and $\H_\alpha$ depend on $\alpha$ and
 that $\cL_\alpha$ is bounded on $\H_\alpha$ for all
 $\alpha \in [0,1]$: observe that $f_{\infty,\alpha} \ge \alpha
 M_{T_\infty}$ and therefore
 \begin{align*}
   \left\| \cL_\alpha h \right\|_{\cH_\alpha} ^2
   & \lesssim \| \sigma
     \|_{L^2(\T)}^2 + \alpha \left\| \frac{\tau}{T_\infty}
     -\sigma\right\|^2 _{L^2(\T)} \left( \int_\R \left(\frac
     {v^2}{T_\infty} - 1\right)^2
     \frac{M_{T_\infty}(v)^2}{f_{\infty,\alpha}} \dd v \right)\\
   & \lesssim \| \sigma
     \|_{L^2(\T)}^2 + \left\| \frac{\tau}{T_\infty}
     -\sigma\right\|^2 _{L^2(\T)} \left( \int_\R \left(\frac
     {v^2}{T_\infty} - 1\right)^2
     M_{T_\infty}(v) \dd v \right)\\
    & \lesssim \| \sigma
      \|_{L^2(\T)}^2 + \frac{1}{T_\infty} \| \tau \|_{L^2(\T)} ^2  
    \lesssim \max\left\{ 1 \ ; \, \frac{1}{T_\infty} \right\} \, \| h
     \|_{\cH_\alpha} ^2. 
 \end{align*}

\subsection{Microscopic coercivity}

We shall now prove that the null space of $\cL$ is the space of
functions $\sigma(x) f_{\infty,\alpha}(v)$ for $\sigma \in L^2(\T)$
and prove a spectral gap on the orthogonal of this null space.

\begin{lemma}\label{mco}
  Let $\alpha \in [0,1]$ and $\cL_\alpha$ defined as in
  \eqref{BGKlin}. Then for all $h\in \cH_\alpha$,
  \begin{equation}\label{diss1}
   \langle h, \cL_\alpha h\rangle_\H  \leq
    -\frac{1-\alpha}{2} \int_{\T\times \R} \left| h(x,v) - f_{\infty,\alpha}(v)
      \int_\R h(x,w)\dd w \right|^2
    \frac{1}{f_{\infty,\alpha}(v)} \dd x \dd v  \ .
\end{equation}
\end{lemma}

\begin{remark} If we had taken $\H_\alpha$ to consist if complex valued functions, 
we would need a real part on the left side of the inequality since ${\mathcal L}_\alpha$ is not self-adjoint. 
  Note that for $\alpha\in [0,1)$, the constant $\lambda_\alpha$ is
  strictly positive but $\lambda_\alpha \to 0$ as $\alpha \to 1$. This
  reflects the fact, see the proof below, that the dissipativity in
  the energy mode is lost because there is energy conservation in this
  limit. In fact in this limit case the microscopic coercivity
  nevertheless holds once accounting for the larger null space of
  $\cL_1$. We are not concerned in this case for which the NESS is
  already known, and we refer to~\cite{AAC} for a study of the
  microscopic coercivity and hypocoercivity for the equation of this
  limit case; This could also be deduced from  the abstract results
  in~\cite{MN}. 
\end{remark}
 
\begin{proof}
  Let us define the following orthonormal family in
  $L^2(f_{\infty,\alpha}^{-1} {\rm d}v)$:
  \begin{equation}\label{lowh}
    H_0(v) := f_{\infty,\alpha}, \quad H_1(v) :=
    \frac{1}{\sqrt{T_\infty}}v f_{\infty,\alpha},
    \quad H_{2}(v) := c_{\alpha}\left(\frac{v^2}{T_\infty} -
      1\right) f_{\infty,\alpha},
  \end{equation}
  where $c_\alpha >0$ is the normalizing constant so that
  $\|H_2\|_{\cH_\alpha}=1$ (one can check that
  $c_{\alpha}^{-2} = 3\left(\alpha + (1-\alpha) \left(2 -
      \frac{T_1T_2}{T_\infty^2}\right)\right)-1$). We also define the
  corresponding orthogonal projections $\Pi_0$, $\Pi_1$, $\Pi_{2}$ in
  $L^2(f_{\infty,\alpha}^{-1} {\rm d}v)$ (note that they all depend on
  $\alpha$):
  \begin{equation}
    \label{eq:projection}
    \begin{cases} \ds
      \Pi_0(h)(v) := \left( \int_\R h(w) \dd w\right) H_0(v), \\[3mm]
      \ds \Pi_1(h)(v)
      := \left( \int_\R h(w) \frac{w}{\sqrt{T_\infty}} \dd w \right) H_1(v),\\[3mm]
      \ds \Pi_2(h)(v)
      := \left( \int_\R h(w) c_\alpha \left( \frac{w^2}{T_\infty} -1 \right)
        \dd w \right) H_2(v) = \left( \frac{\tau}{T_\infty} - \sigma
      \right) c_\alpha H_2(v).
    \end{cases}
  \end{equation}
  Finally we denote $\Pi^\bot$ the orthogonal projection on
  $\{ H_0, H_2 \}^\bot$ (note that this projection includes $\Pi_1$ in
  its range).

  The linearized collision operator $\cL_\alpha$ can be written using this notation a
  \begin{equation}\label{eq:lin-proj}
    \cL_\alpha h = \Pi_0(h) 
    + \frac{\alpha}{2 c^2_{\alpha}} \Pi_2(h)
    \frac{M_{T_\infty}}{f_{\infty,\alpha}} -h. 
  \end{equation}

  We then compute the Dirichlet form
  \begin{align*}
    \langle \cL_\alpha h, h \rangle_{\cH_\alpha} =
    & \| \Pi_0(h) \|^2
      _{\cH_\alpha} + \frac{\alpha}{2 c^2_{\alpha}} \left\langle  \Pi_2(h)
      \frac{M_{T_\infty}}{f_{\infty,\alpha}}, \Pi_2(h)
      \right\rangle_{\cH_\alpha} \\
    & + \frac{\alpha}{2 c^2_{\alpha}} \left\langle  \Pi_2(h)
      \frac{M_{T_\infty}}{f_{\infty,\alpha}}, \Pi^\bot(h)
      \right\rangle_{\cH_\alpha} - \| h \|_{\cH_\alpha} ^2
  \end{align*}
  where we have used that $H_2 M_{T_\infty} f_{\infty,\alpha}^{-1}$ is
  orthogonal to $H_0$ and $H_1$ in $L^2(f_{\infty,\alpha}^{-1} {\rm
    d}v)$. 

  Let us define
  \begin{equation}\label{Uadef}
    \cU_\alpha(v) := H_2(v) \frac{M_{T_\infty}(v)}{f_{\infty,\alpha}(v)}. 
  \end{equation}
  The projection of this function on $H_0$ and $H_1$ in
  $L^2(f_{\infty,\alpha}^{-1} {\rm d}v)$ is zero and its projection on
  $H_2$ has coefficient
  \begin{equation*}
    \int_\R \cU_\alpha(v) H_2(v) \frac{1}{f_{\infty,\alpha}(v)} \dd v
    = c_\alpha^2 \int_\R \left( \frac{v^2}{T_\infty} - 1 \right) ^2
    M_{T_\infty}(v) \dd v = 2 c_\alpha ^2.
  \end{equation*}
  Its norm satisfies
  \begin{align*}
    \int_\R \cU_\alpha(v)^2 \frac{1}{f_{\infty,\alpha}(v)} \dd v
    & = c_\alpha^2 \int_\R \left( \frac{v^2}{T_\infty} - 1 \right) ^2
      \frac{M_{T_\infty}(v)^2}{f_{\infty,\alpha}(v)} \dd v \\
    & \le \frac{c_\alpha^2}{\alpha} \int_\R \left( \frac{v^2}{T_\infty} - 1 \right) ^2
      M_{T_\infty}(v) \dd v \le \frac{2 c_\alpha^2}{\alpha}
  \end{align*}
  where we have used in the last line
  $f_{\infty,\alpha} \ge \alpha M_{T_\infty}$. We then decompose
  orthogonally $\cU_\alpha = \Pi_2(\cU_\alpha) + \Pi^\bot(\cU_\alpha)$
  and deduce by Pythagoras' theorem that
  \begin{equation}\label{eq:Ubot}
    \left\| \Pi^\bot (\cU_\alpha)
    \right\|_{L^2(f_{\infty,\alpha}^{-1})} ^2 \le 4 c_\alpha^4
    \frac{1-\alpha^2}{\alpha^2}. 
  \end{equation}

  We deduce on the one hand that 
  \begin{align*}
     \frac{\alpha}{2 c^2_{\alpha}} \left\langle  \Pi_2(h)
    \frac{M_{T_\infty}}{f_{\infty,\alpha}}, \Pi_2(h)
      \right\rangle_{\cH_\alpha} 
    & = \frac{\alpha}{2} \int_{\T \times \R}  \left( \frac{\tau}{T_\infty} - \sigma
    \right)^2 \cU_\alpha(v) H_2(v) \dd x \dd v \\
    & = \alpha c_\alpha^2  \int_{\T}  \left( \frac{\tau}{T_\infty} - \sigma
    \right)^2 \dd x = \alpha  \| \Pi_2(h) \|_{\cH_\alpha}^2. 
  \end{align*}

  We deduce on the other hand
  \begin{align*}
    & \frac{\alpha}{2 c^2_{\alpha}} \left\langle  \Pi_2(h)
      \frac{M_{T_\infty}}{f_{\infty,\alpha}}, \Pi^\bot(h)
    \right\rangle_{\cH_\alpha} \\
    & = \frac{\alpha}{2c_\alpha} \int_{\T \times \R}  \left( \frac{\tau}{T_\infty} - \sigma
    \right) \Pi^\bot(\cU_\alpha) \Pi^\bot(h) \dd x \dd v\\
    & \le \frac{\alpha}{2 c_\alpha} \left\| \Pi^\bot (\cU_\alpha)
    \right\|_{L^2(f_{\infty,\alpha}^{-1})}  \left( \int_\T \left( \frac{\tau}{T_\infty} - \sigma
    \right)^2 \dd x \right)^{1/2}
      \| \Pi^\bot(h) \|_{\cH_\alpha} \\
    & \le \frac{\alpha}{2 c_\alpha^2} \left\| \Pi^\bot (\cU_\alpha)
    \right\|_{L^2(f_{\infty,\alpha}^{-1})} \| \Pi_2(h) \|_{\cH_\alpha}
      \| \Pi^\bot(h) \|_{\cH_\alpha} \\
    & \le (1-\alpha^2)^{1/2} \| \Pi_2(h) \|_{\cH_\alpha}
      \| \Pi^\bot(h) \|_{\cH_\alpha}
  \end{align*}
  where we have used~\eqref{eq:Ubot} in the last line.

  We therefore obtain
    \begin{equation*}
      \langle \cL_\alpha h, h \rangle_{\cH_\alpha} \le
      (\alpha-1) \| \Pi_2(h) \|_{\cH_\alpha}^2 + (1-\alpha^2)^{1/2} \|  \Pi_2(h)
      \|_{\cH_\alpha} \| \Pi^\bot(h)
      \|_{\cH_\alpha} - \| \Pi^\bot h \|_{\cH_\alpha} ^2.
    \end{equation*}
    The quadratic form on the right hand side is negative for
    $\alpha \in [0,1)$ since then $(1-\alpha^2) < 4
    (1-\alpha)^2$. It degenerates at $\alpha=1$. The matrix of the
    quadratic form
    \[
      (V_1,V_2) \in \R^2 \mapsto (\alpha-1)V_1 - V_2 + (1-\alpha^2)^{1/2} V_1 V_2
    \]
    has characteristic polynomial
    \[
      P(X) = X^2 +(2-\alpha) X + \frac14 (1-\alpha)(3-\alpha)
    \]
    whose roots are $\alpha/2+1\pm 1/2$. The greatest eigenvalue is
    therefore $(\alpha-1)/2$ which concludes the proof.
  \end{proof}

\subsection{Hypocoercivity}

With the microcoercivity at hand we can now readily prove 
 prove hypocoercivity: That is, we shall prove that for some constant $C< \infty $ and some $\lambda> 0$,
and solution $h_t$  in $\H_\alpha$ of our linearized evolution equation satisfies 
\begin{equation}\label{hyp}
\|h_t\|_{\H_\alpha} \leq  Ce^{-\lambda t}\|h_0\|_{\H_\alpha}\ .
\end{equation}
With this in hand, it is a simple matter to prove the nonlinear stability. We discuss two approaches to proving \eqref{hyp} for our model. One approach applies when the steady state is symmetric in $v$, as in our case. As noted  in \cite{AAC},
whenever this is the case, there is a natural orthonormal basis such that in this basis the streaming operator is represented by uncoupled tridiagonal blocks, while the gain term in ${\mathcal L}_\alpha$ is represented by uncoupled lower triangular blocks. 
This structure permits the extraction of simple, explicit bounds on $\mu$.  Another approach, developed in \cite{DMS} is more abstract and not requiring symmetry of the steady state, provides an efficient route to \eqref{hyp}.
In this section we prove:

\begin{theorem}\label{expl}  The decay estimate \eqref{hyp} is valid with the following explicit values of $\lambda$ and $C$:

\smallskip
\noindent{\it (1)} If $c_\alpha^2\frac{1-\alpha}{2\sqrt{T_\infty}} < \frac12$, we may take $C =4$ and $\lambda = \frac{1-\alpha}{8}$. 

\smallskip
\noindent{\it (2)} $c_\alpha^2\frac{1-\alpha}{2\sqrt{T_\infty}} \geq  \frac12$, we may take $C =4$ and $\lambda = \frac{\sqrt{T_\infty}}{8}$. 

This is also true, with the same constants, if we replace $\H_\alpha$ by $\cH_\alpha^1$, the latter Hilbert space being defined in \eqref{sob}. 
\end{theorem}

\begin{remark}  This result is stronger than Theorem~\ref{stab} in that it provides explicit bounds on the exponential rate of convergence.
The reason for the $\sqrt{T_\infty}$ dependence of $\lambda$ for small $T_\infty$ is that hypocoercivity depends on the effects of the 
streaming operator $v\partial_x$ to ``mix'' the conserved mass mode into the dissipated modes, and the tridiagonal representation of the streaming operator  given in \eqref{L1L2} shows that its mixing effects are proportional to $\sqrt{T_\infty}$.  When $T_\infty$ is large, there is rapid mixing, but this can do only so much good: The dissipativity of the non-conserved modes as estimated in Lemma \eqref{mco} is bounded independent of $T_\infty$, but the mixing only shares this dissipativity around, it cannot improve the dissipativity no matter how fast it runs. 

It is also worth noting that we have simple bounds on $c_\alpha$: By the arithmetic-geometric mean inequality, $0 \leq 
\frac{T_1T_2}{T^2} \leq 1$, and hence
$2 \leq c_{\alpha,T}^{-2} \leq  3(2-\alpha) -1$.

\end{remark}

The proof of Theorem~\ref{expl} is quite short once one has computed matrix representations of $v\partial_x$ and ${\mathcal L}_\alpha$ with respect to a basis that we now introduce:
The basis is
$\{e^{ikx}g_m(v)\}_{m\geq 0,k\in \Z}$ where
$\{g_m(v)\}_{m\geq 0}$ is the sequence one gets by applying this
Gram-Schmidt orthonormalization procedure to the sequence of functions
$v^m f_{\infty,\alpha}(v)$ for $m\geq 0$.  We write these in the form
$g_m(v) = H_m(v)f_{\infty,\alpha}(v)$
where $H_m$ is a polynomial of degree $m$.  
For $h\in \H_\alpha$, we write 
\begin{equation}\label{expand}
h = \sum_{m\geq 0,k\in \Z} ik \widehat h_m(k) e^{ikx}g_m(v)\ .
\end{equation}

   The  action of the free streaming operator $\cS := v\partial_x$  then is  
   $$
   \cS h(x,v) = \sum_{m\geq 0,k\in \Z} ik \widehat h_m(k)
   e^{ikx}vg_m(v)\ .
   $$
   It is a simple consequence of the fact that $f_{\infty,\alpha}$ is even in $v$ that  for each $m \geq 1$, $vg_m$ is a linear combination of $g_{m-1}$ and $g_{m+1}$; see \cite{AAC}.
Since the operation of multiplication
by $v$ is self-adjoint, it follows that that there exist numbers
$\{a_n\}_{m\geq 1}$ such that
\[
  vg_m(v) =  \sqrt{T_\infty}a_{m+1} g_{m+1}(v) +   \sqrt{T_\infty}a_m g_{m-1}(v)
  \ \mbox{ with the convention } g_{-1} := 0.
\]

For $m=0$, by \eqref{lowh},
$vg_0(v) = vf_{\infty,\alpha} = \sqrt{T_\infty}g_1$, and hence $a_1 = 1$.
Likewise,
\[
  vg_1(v) = \frac{1}{\sqrt{T_\infty}} v^2 f_{\infty,\alpha}
  =  {\sqrt{T_\infty}} \left(\frac{v^2}{T_\infty} - 1\right) f_{\infty,\alpha} +
  \sqrt{T_\infty}g_0(v)
\]
and hence $a_2 = c_{\alpha}^{-1}$. One can work out $a_m$ for or
higher values of $m$, but these are not needed here.

Let $\widehat {\bf h}(k)$ denote the element of $\ell_2$ whose $m$th
component is $\widehat h_m(k)$. Then the corresponding vector of
coefficients for $\cS h$ is given by $ik{\bf S} \widehat {\bf h}(k)$
where ${\bf S}$ is the tri-diagonal matrix
 \begin{equation}\label{L1L2}
  {\bf S}= \sqrt{T_\infty} \left(\begin{array}{ccccc}
  0 &  a_1 &  0 &  0 &  \cdots \\
  a_1 & 0 &  a_2 &  0 &  \cdots \\
  0 &  a_2 & 0 &  a_3  &  \cdots\\
 0 &  0 &  a_3 & 0 &  \cdots\\
 \vdots & \vdots & \vdots & \vdots & \ddots 
  \end{array}\right)   = 
  \sqrt{T_\infty} \left(\begin{array}{ccccc}
  0 & 1 &  0 & 0& \cdots \\
  1 & 0 &  c_{\alpha}^{-1} &  0 & \cdots \\
  0 &  c_{\alpha}^{-1}& 0 &  a_3 & \cdots \\
  0 &  0 &  a_3 &  0 & \cdots\\
   \vdots & \vdots & \vdots & \vdots & \ddots 
  \end{array}\right).
\end{equation}

We now turn to the gain term in the  linearized collision operator.
The linearized collision operator does not act on the spatial 
variable. By projecting~\eqref{eq:lin-proj} in our basis we
get
\begin{equation}\label{preex}
  \widehat{\cL_\alpha h}(k,v)   =  \widehat{h}_0(k)g_0(v)
  + \frac{\alpha}{2 c^2_{\alpha}} \widehat{h}_2(k) H_2(v)
  M_{T_\infty}(v) - \widehat{h}(k,v).
\end{equation}

For each $k$, the action of $\cL_\alpha$ in the $\{g_n\}_{m\geq 0 }$
basis is given by
\[
  \widehat{\cL_\alpha h}_m(k) = ({\bf L}_\alpha \widehat{\bf h}(k))_m
\]
where $\LL_\alpha + {\bf Id}$ is the matrix whose first column is the
unit vector $(1,0,0,\dots)$, whose third column is the vector
$(0,0,\alpha, b_3,b_4, \dots )$, with all other columns being zero, and with
$$
b_m :=  \frac{\alpha}{2 c^2_{\alpha}} \widehat{h}_2(k) \int H_2(v) H_m(v)
  M_{T_\infty}(v) {\rm d}v\ , 
$$
so that, in particular, $b_0 =b_1 = 0$ and $b_2 = \alpha$.

Therefore, we may rewrite our linearized equation as  the decoupled system of equations
\begin{equation}\label{decooup}
\partial_t \widehat{\bf h}(k) = (\LL_\alpha -ik{\bf S})\widehat{\bf h}(k)\ ,
\end{equation}
for each $k\in \Z$. For $k=0$ we simply have from~\eqref{preex}
$\partial_t \widehat{\bf h}(0)  = -\widehat{\bf h}(0)$
since $\widehat h(0)_0 =  \widehat h(0)_2 = 0$. 

For each $k\neq 0$, define $\CC_k = -(\LL_\alpha -ik{\bf S})$.  We
seek a positive definite matrix $\PP_k$ such that for some fixed
$\lambda>0$,
\begin{equation}\label{modno}
  \fa k \in \Z, \quad \CC_k^* \PP_k + \PP_k\CC_k \geq 2\lambda \PP_k.
\end{equation}

Then if we define the Lyapunov function $e(h)$ by
$e(h) = \sum_{k\in \Z} \omega_k \left\langle \widehat{\bf h}(k), \PP_k
    \widehat{\bf h}(k)\right\rangle_{\ell_2}$,
for {\em any} sequence $\{\omega_k\}_{k\in \Z}$ of positive numbers we
have that for any solution of our linearized equation with initial
data $h_0$ with $e(h) < \infty$,
\begin{equation}\label{decay}
  \frac{{\rm d}}{{\rm d}t} e(h(t))  \leq -2\lambda e(h(t)) \quad
  \Longrightarrow \quad e(h(t))  \leq e^{-2t\lambda} e(h_0).
\end{equation}
We will construct the matrices $\PP_k$ so that for some
$C>0$
\begin{equation}\label{decay}
  \fa k \in \Z\backslash\{0\}, \quad K {\bf I}   \leq \PP_k  <\frac{1}{K}{\bf I}.
\end{equation}
This implies  that the function $e(h)$ is equivalent to the norm on $\H_\alpha$ if we take each $\omega_k=1$. We then conclude that \eqref{hyp} is valid with  $C = \frac1K$ and the value of $\lambda$ appearing in \eqref{modno}.
By making other choices for $\omega_k$,  we obtain decay in various Sobolev type norms. For instance, taking $\omega_k = 1 + k^2$, we would obtain  decay in the norm defined in \eqref{sob}. 
With the matrix representation computed, and Lemma~\ref{mco} at our disposal, we are ready to prove Theorem~\ref{expl}

\begin{proof}[Proof of Theorem~\ref{expl}]
Because ${\mathcal L}_\alpha$ is lower triangular  with positive  diagonal entries that are uniformly bounded away from zero except of course for the zero in the upper left, and because 
${\bf S}$ is tridiagonal, a simple prescription from \cite{AAC} provides $\PP_k$.
For a parameter $c\in (0,1)$ to be chosen later, define
$\PP_k(a)$ by entering
\begin{equation}\label{Ptwo}
  \left(\begin{array}{cc} 1  & -i c/k   \\ 
          i c/k & 1  \end{array}\right)\  
\end{equation}
as its upper-left $2\times 2$ block, with all other entries being
those of the identity. The eigenvalues of $\PP_k(a)$ are $1$, $1+c/k$
and $1-c/k$, and hence \eqref{decay} is satisfied with $K = 1-c$.

Then simple computations show that 
$$\CC_k^*\PP_k(a) + \PP_k(a)\CC_k  = (2{\bf I} - \LL - \LL^T) + {\bf R}$$
where ${\bf R}$ is the matrix whose upper $3\times 3$ block is 
\begin{equation}\label{block}
\left(\begin{array}{ccc}
2 c\sqrt{T_\infty}  & -i c/k & c\sqrt{T_\infty} c_\alpha^{-1}  \\
i c/k &  0 &0 \\
c\sqrt{T_\infty} c_\alpha^{-1}   & 0 & 0   \\
  \end{array}\right) \,
\end{equation}
and whose remaining entries are all zero. Hence for any (real)  ${\bf h} = (h_0,h_1,h_2\dots)\in \ell^2$, 
$$\langle {\bf h}, {\bf R}{\bf h}\rangle_{\ell^2} =  2 c\sqrt{T_\infty} h_0^2+ 2 c\sqrt{T_\infty} c_\alpha^{-1}h_0h_2 \geq c\sqrt{T_\infty}(h_0^2 - c_\alpha^{-2}h_2^2)\ .$$
By Lemma \eqref{mco},
$$\langle {\bf h},(2{\bf I} - \LL - \LL^T){\bf h}\rangle_{\ell_2} \geq (1-\alpha)\sum_{m=1}^\infty h_m^2\ .$$
Combining these estimates with ${\bf I} \geq (1-c)\PP_k$, \eqref{modno} then holds with
$$2\lambda = \sqrt{T_\infty}\min\left\{ c \ ,\ \frac{1-\alpha}{\sqrt{T_\infty}} - c c_\alpha^{-2}\right\}(1-c)\ .
$$
We now choose $c$ so that $c =\min\{\frac12,  c_\alpha^2\frac{1-\alpha}{2\sqrt{T_\infty}}\}$.
If $c_\alpha^2\frac{1-\alpha}{2\sqrt{T_\infty}} < \frac12$, then 
$\lambda \geq  \frac{1-\alpha}{8}$. 
If $c_\alpha^2\frac{1-\alpha}{2\sqrt{T_\infty}} \geq \frac12$, then 
$\lambda \geq  \frac{\sqrt{T_\infty}}{8}$. 
\end{proof}

 We  now explain another route to \eqref{hyp} which is relies on general abstract results obtained in \cite{DMS}.:

\begin{theorem}[Abstract result of hypocoercivity from~\cite{DMS}]
  \label{theo:hypo}
\label{theo:abstract}
Consider a Hilbert space $\cH$ and two closed unbounded operators
$\cS$ and $\cL$ such that:

\begin{itemize}
\item[{\bf (H1)}] \textbf{Microscopic coercivity}: There exists an
  orthogonal projection $\Pi_0$ such that
  $\cL \Pi_0 = \Pi_0 \cL =0$ and there is $\lambda_m > 0$ such that
\[
  \fa h \in \mbox{\emph{Domain}}(\cL), \quad -\langle \cL h,h \rangle \ge
  \lambda_m\,\|(I-\Pi_0)\,h\|^2.
\]

\item[{\bf (H2)}] \textbf{Macroscopic coercivity}: The operator $\cS$
  satisfies $\cS^*= -\cS$ (skew symmetry) and there exists
  $\lambda_M > 0$ such that
\[
\| \cS \Pi_0 h\|^2 \ge \lambda_M\,\|\Pi_0 h\|^2\quad\mbox{\emph{for
    all}}\; h\in \cH \;\mbox{\emph{such that}}\;\Pi_0(h) \in \mbox{\emph{Domain}}(\cS).
\]

\item[{\bf (H3)}] \textbf{Consistency}: $\Pi_0 \cS \Pi_0 = 0$.

\item[{\bf (H4)}] \textbf{Auxiliary operator}: Define
  $\cA := (1 + (\cS\Pi_0)^*(\cS\Pi_0))^{-1} (\cS\Pi_0)^*$ and assume that
  $\cA \cS (1-\Pi_0)$ and $\cA\cL$ are bounded with a constant $C_M>0$
  such that
  \[
    \fa h \in \cH, \quad \|\cA \cS (1-\Pi_0)h\| + \|\cA\cL h\| \le C_M\,\|(1-\Pi_0)h\|.
  \]
\end{itemize}

Then there exist positive constants $\lambda>0$ and $C>0$, which are
explicitly computable in terms of $\lambda_m$, $\lambda_M$, and $C_M$,
such that
\[
\label{eq:decH}
\fa h \in \cH, \ \fa t \ge 0, \quad \big\|e^{t(\cL-\cS)} h\big\| \le C e^{-\lambda t}\,\|h\|.
\]
\end{theorem}

We  give a
short proof of  here for the convenience of the reader.

\begin{proof}[Proof of Theorem~\ref{theo:hypo}]
Define the modified norm
\[
\cH[h] := \frac 12\,\|h\|^2 + \eps\,\langle \cA h, h \rangle
\]
where $\eps>0$ will be chosen small enough below. Given $h_t :=
e^{t(\cL-\cS)}h$, we compute 
\begin{align*}
  \frac{{\rm d}}{{\rm d}t}\,\cH[h_t] = & \langle \cL h_t,h_t \rangle
  - \eps \langle\cA\cS\Pi_0 h_t,h_t \rangle -
  \eps \langle \cA\cS(1-\Pi_0)h_t,h_t \rangle\\
  & + \eps \langle \cS\cA
  h_t,h_t \rangle + \eps\langle\cA\cL h_t,h_t\rangle \\
  =:& -\cD[h_t].
\end{align*}
We have used here that $\cL^* \cA =0$ which follows from
$\cA = \Pi_0 \cA$ and $\Pi_0 \cL =0$ in {\bf (H1)}. By {\bf (H1)}, {\bf
  (H2)}, and by
$\cA\cS\Pi_0 = (1 + (\cS\Pi_0)^*(\cS\Pi_0))^{-1} (\cS\Pi_0)^*(\cS\Pi_0)$, the
sum of the first two terms in $\cD[h_t]$ is coercive:
\[
  -\langle \cL h_t,h_t \rangle + \eps \langle\cA\cS\Pi_0 h_t,h_t\rangle
  \ge \min\Big\{ \lambda_m, \frac{\eps \lambda_M}{1+\lambda_M}
  \Big\}\,\|h_t\|^2.
\]

Let us prove that the operators $\cA$ and $\cS\cA$ are bounded:
\begin{equation}
  \label{ExplicitBound}
  \fa h \in \cH, \quad \|\cA h\| \le \frac12 \|(1-\Pi_0)h\| \quad\mbox{and}\quad \| \cS\cA h\| \le
  \|(1-\Pi_0)h\|.
\end{equation}

The equation $\cA h = g = \Pi_0 g$ (remember that $\cA = \Pi_0 \cA$) is
equivalent to
\[\label{Af=g}
(\cS\Pi_0)^* h = g + (\cS\Pi_0)^*(\cS\Pi_0)g.
\]
Taking the scalar product of the
above equality with $g$ and using {\bf (H3)}, we get
\begin{eqnarray*}
  \|g\|^2 + \|\cS\Pi_0 g\|^2
  &=& \langle h,\cS\Pi_0 g \rangle = \langle(1-\Pi_0)h,\cS\Pi_0 g \rangle \\
  &\le& \|(1-\Pi_0)h\| \|\cS\Pi_0 g\| \le \frac14 \|(1-\Pi_0)h\|^2
        + \|\cS\Pi_0 g\|^2,
\end{eqnarray*}
which completes the proof of~\eqref{ExplicitBound}.

The first inequality in \eqref{ExplicitBound} implies that $\H[h]$ is
equivalent to $\|h\|^2$:
\begin{equation}\label{equivalence}
  \frac 12 (1-\eps) \|h\|^2
  \le \H[h] \le \frac 12 (1+\eps) \|h\|^2.
\end{equation}
The second inequality in \eqref{ExplicitBound} and {\bf
  (H1)-(H2)-(H3)-(H4)} imply
\begin{eqnarray*}
  \cD[f]
  &\kern -3pt\ge
  & \kern -3pt \lambda_m \|(1-\Pi_0)h_t\|^2 +
    \frac{\eps\lambda_M}{1+\lambda_M} \|\Pi_0 h_t\|^2 - \eps (1+C_M)
    \|(1-\Pi_0)h_t\| \|h_t\| \\
  &\kern -3pt\ge
  &\kern -3pt \left[\lambda_m - \eps(1+C_M)\frac{1}{2\delta}\right]
    \|(1-\Pi_0)h_t\|^2
    + \eps\left[\frac{\lambda_M}{1+\lambda_M} -
    (1+C_M)\frac{\delta}{2}\right] \|\Pi_0 h_t\|^2
\end{eqnarray*}
for an arbitrary $\delta>0$. By choosing first $\delta$ and then
$\eps$ small enough, a positive constant $\kappa$ can be found, such
that $\cD[h_t]\ge \kappa\|h_t\|^2$. Using \eqref{equivalence}, this
implies
\[
\frac{{\rm d}}{{\rm d}t}\,\cH[h_t] \le -\frac{2\kappa}{1+\eps}\cH[h_t],
\]
completing the proof with $\lambda = \kappa/(1+\eps)$ and
$C = \sqrt{1+\eps}/\sqrt{1-\eps}$.
\end{proof}

\begin{proof}[Second proof of Theorem~\ref{stab}]
  We consider $\alpha \in [0,1)$ (for $\alpha=1$
  the result is already known from \cite{AAC}). We apply the previous
  Theorem~\ref{theo:hypo} with $\cH$ being the subspace  of
  $\cH_\alpha = L^2(f_{\infty,\alpha}^{-1}{\rm d}x{\rm d}v)$ consisting of functions that satisfy
   the zero global mass condition
  $\int_{\T \times \R} h(x,v) \dd x \dd v=0$. We take 
    $\cS = v \partial_x$  and $\cL= \cL_\alpha$,
  the linearized operator defined in~\eqref{eq:lin-proj}, and
  $\Pi_0(h) = (\int_\R h \dd v)f_{\infty,\alpha}$ defined
  in~\eqref{eq:projection}. Then {\bf (H1)} is proved in
  Lemma~\ref{mco} and {\bf (H2)} follows from
  \[
    \int_{\T \times \R} \left( v \partial_x \sigma f_{\infty,\alpha}
    \right)^2 \frac{\dd x \dd v}{f_{\infty,\alpha}} = T_\infty \|
    \partial_x \sigma \|^2 _{L^2(\T)} \ge T_\infty \| \sigma \|_{L^2(\T)}
  \]
  where we have used  the Poincar\'e inequality in the
  unit torus, and the zero global mass condition. {\bf (H3)} follows
  from $\int_\R v f_{\infty,\alpha}(v) \dd v =0$. Finally, to prove
  {\bf (H4)}, we first figure out some explicit formula for $\cA$:
  \[
    \begin{cases}
      \ds \cA h = - \left[\left( 1-T_\infty \partial_x ^2 \right)^{-1}
        \partial_x j \right] f_{\infty,\alpha}, & \ds \quad j(x) := \int_\R
      v h(x,v) \dd v \\[3mm]
      \ds \cA \cS h = - \left[\left( 1-T_\infty \partial_x ^2 \right)^{-1}
        \partial_x^2 \tau \right] f_{\infty,\alpha}, & \ds \quad \tau(x) := \int_\R
      v^2 h(x,v) \dd v.
    \end{cases}
  \]
  Since $\left( 1-T_\infty \partial_x ^2 \right)^{-1} \partial_x$ and
  $\left( 1-T_\infty \partial_x ^2 \right)^{-1} \partial_x^2$ are
  bounded operators on $L^2(\T)$, we deduce that $\cA$ and $\cA \cS$
  are bounded, and {\bf (H4)} follows since $\cL_\alpha$ is bounded
  and $\cL_\alpha \Pi_0=0$. This concludes the proof of hypocoercivity
  for our linearized operator $\cL_\alpha -\cS$.
  
  \subsection{Nonlinear stability} We close this section by proving
nonlinear stability. Let
  $\alpha \in [0,1)$ and consider a probability density
  $f_{in} \in \cH_\alpha$ and define
  $h_{in} := f_{in} - f_{\infty,\alpha} \in \cH_\alpha$. This
  fluctuation has zero global mass by definition. We define the
  solution through
  \begin{equation*}
    h_t = e^{t(\cL_\alpha -\cS)} h_{in} + \int_0 ^t
    e^{(t-s)(\cL_\alpha-\cS)} \cR[h_s] \dd s
  \end{equation*}
  with the nonlinear remainder term defined by 
  \[
   \cR[h] :=  \M_f - \M_{f_{\infty,\alpha}}
   -\left[  M_{T_\infty}\sigma
     + \frac{1}{2}  \left(\frac {v^2}{T_\infty} - 1\right)
     M_{T_\infty}(v)\left(\frac{1}{T_\infty} \tau -
       \sigma\right) \right]
 \]
 with $\sigma$ and $\tau$ defined in terms of $h$ as before. Taylor
 expansions and straightforward calculations, using the multiplicative property of the $H^1(\T)$ Sobolev norm,  show that
 \[
   \| \cR[h] \|_{\cH_\alpha^1} \lesssim \| \sigma\|^2_{H^1(\T)} + \|
   \tau \|^2 _{H^1(\T)} \lesssim \| h\|^2_{\cH_\alpha^1}. 
 \]
(For more detail in a closely related argument, see \cite{AAC}.) We deduce the a priori estimate
 \[
   \| h_t\|_{\cH_\alpha^1} \le C e^{-\lambda t} \| h_{in}\|_{\cH_\alpha^1}
   + C' \int_0 ^t e^{-\lambda (t-s)} \| h_s \|_{\cH_\alpha^1} ^2 \dd s
 \]
 for some constants $C,C',\lambda >0$, and one can use the values of $C$ and $\lambda$ provided by Theorem~\ref{expl}. By a standard argument, this shows global existence and
 exponential decay at rate $\lambda$  when $\| h_{in} \|_{\cH_\alpha^1}$ is small enough. 
\end{proof}

\bigskip

\noindent
{\bf Acknowledgements.} E.~C. acknowledges partial support from NSF grant DMS--174625.  J.~L. acknowledges support from the AFOS under the award number  FA-9500-16-0037
C.~M.  acknowledges partial support
from the ERC grant MAFRAN. He also thanks Rutgers University and the
IAS for their invitations, during which this work was started.

\bigskip
\signec
\signre
\signjl
\signrm
\signcm

\end{document}